\documentclass[UKenglish,cleveref, autoref, thm-restate]{lipics-v2021}

\hideLIPIcs  

\graphicspath{{./graphics/}}

\title{Improved Approximation Algorithms and Hardness Results for Shortest Common Superstring with Reverse Complements} 

\titlerunning{Improved Approximation Algorithms and Hardness Results for SCS-RC} 

\author{Ryosuke Yamano}{Department of Computer Science, Graduate School of Information Science and Technology, The University of Tokyo, Japan \and Division of Medical Data Informatics, Human Genome Center, Institute of Medical Science, The University of Tokyo, Japan}{ryoyamano15@g.ecc.u-tokyo.ac.jp}{https://orcid.org/0009-0002-1683-5179}{}
\author{Tetsuo Shibuya}{Division of Medical Data Informatics, Human Genome Center, Institute of Medical Science, The University of Tokyo, Japan}
{tshibuya@hgc.jp}{https://orcid.org/0000-0003-1514-5766}{}

\authorrunning{R. Yamano and T. Shibuya}

\Copyright{R. Yamano and T. Shibuya} 



\relatedversion{} 

\funding{This work was supported by MEXT KAKENHI Grant Numbers 21H05052, 23H03345, and 23K18501.}

\nolinenumbers 

\newcommand{\op}[1]{\mathrm{#1}}

\usepackage[ruled]{algorithm}
\usepackage{subcaption}

\EventEditors{John Q. Open and Joan R. Access}
\EventNoEds{2}
\EventLongTitle{42nd Conference on Very Important Topics (CVIT 2016)}
\EventShortTitle{CVIT 2016}
\EventAcronym{CVIT}
\EventYear{2016}
\EventDate{December 24--27, 2016}
\EventLocation{Little Whinging, United Kingdom}
\EventLogo{}
\SeriesVolume{42}
\ArticleNo{23}

\begin{document}

\maketitle

\begin{abstract}
The Shortest Common Superstring (SCS) problem is a fundamental task in sequence analysis.
In genome assembly, however, the double-stranded nature of DNA implies that each fragment may occur either in its original orientation or as its reverse complement. This motivates the Shortest Common Superstring with Reverse Complements (SCS-RC) problem, which asks for a shortest string that contains, for each input string, either the string itself or its reverse complement as a substring.
The previously best-known approximation ratio for SCS-RC was $\frac{23}{8}$.
In this paper, we present a new approximation algorithm achieving an improved ratio of $\frac{8}{3}$. Our approach computes an optimal constrained cycle cover by reducing the problem, via a novel gadget construction, to a maximum-weight perfect matching in a general graph.
We also investigate the computational hardness of SCS-RC. While the decision version is known to be NP-complete, no explicit inapproximability results were previously established. We show that the hardness of SCS carries over to SCS-RC through a polynomial-time reduction, implying that it is NP-hard to approximate SCS-RC within a factor better than $\frac{333}{332}$. Notably, this hardness result holds even for the DNA alphabet.
\end{abstract}

\section{Introduction}
The Shortest Common Superstring (SCS) problem is to find a shortest string that contains every string in a given set as a substring.
It is a classical problem in combinatorial optimization and has numerous applications across a wide range of fields \cite{Gevezes2014}. 
A particularly important application arises in DNA assembly. A DNA molecule consists of four nucleotides, namely adenine, thymine, guanine, and cytosine. 
The sequencing task aims to reconstruct the original molecule from a collection of short reads, which can be naturally viewed as an instance of the SCS problem over a quaternary alphabet.

However, the standard SCS formulation does not capture a crucial aspect of genome assembly, namely the double-stranded nature of DNA.
Each read may originate from either strand of the DNA molecule and may therefore appear in its original orientation or as its reverse complement.
Since sequencing technologies typically do not reveal the strand of origin, the target genome must contain, for each read, either the read itself or its reverse complement as a substring \cite{MyersJr+2016+126+132}.
This observation motivates the \emph{Shortest Common Superstring with Reverse Complements} (SCS-RC) problem, which seeks a shortest string that contains, for each input string, either the string itself or its reverse complement.

Over the past decades, the standard SCS problem has been extensively studied, resulting in a rich body of literature, including a wide range of approximation algorithms \cite{Blum.et.al, CPM1996.2+2/3approx, Breslaure.1997.OverlapRotationLemma, SIAM.2+1/2approx, KAPLAN200513.greedy3.5n, SODA13.2+11/23approx, STOC22.ImprovedApproximateGuarantees, englert_et_al:LIPIcs.ISAAC.2023.29} as well as several inapproximability results \cite{1217/1216.inapprox.MFCS,333/332.inapprox.CATS2013}. 
The currently best-known approximation ratio guarantee is $\frac{\sqrt{67} + 14}{9} \approx 2.465$, due to Englert et al. \cite{englert_et_al:LIPIcs.ISAAC.2023.29}, whereas the strongest inapproximability result establishes that it is NP-hard to approximate SCS within a factor better than $\frac{333}{332}$ \cite{333/332.inapprox.CATS2013}.

In contrast, despite its clear relevance to genome assembly, the SCS-RC problem has received considerably less attention. 
Jiang et al.\ \cite{JIANG1992195} extended the work of Blum et al.\ \cite{Blum.et.al} and presented a 3-approximation algorithm for SCS-RC. 
Yamano and Shibuya \cite{YamanoShibuya2026SCSRC} subsequently improved the approximation ratio to $\frac{23}{8} = 2.875$, which remains the best known. 
On the hardness side, Kececioglu \cite{kececioglu} proved the NP-completeness of the decision version of SCS-RC, which is, to the best of our knowledge, the only known hardness result. 
These results highlight a substantial gap between the current understanding of SCS and SCS-RC. 
In this work, we address this gap by improving both the approximation ratio guarantees and the hardness results for SCS-RC. 

\subsection{Our Contributions} \label{sec: Our Contributions}
We first establish \cref{thm: 2.67 approximation},
improving the approximation ratio for SCS-RC
from $\frac{23}{8} = 2.875$ to $\frac{8}{3} \approx 2.67$.
Our result extends the classical $\frac{8}{3}$-approximation algorithm
for SCS due to Breslauer et al.~\cite{Breslaure.1997.OverlapRotationLemma},
which is based on computing a minimum-weight cycle cover,
that is, a collection of vertex-disjoint directed cycles
covering all vertices of a weighted directed graph
encoding string overlaps in the SCS setting.
In our setting, this notion is modified as follows:
instead of requiring the cycle cover to include all vertices,
we require that exactly one vertex is selected from each complementary pair.

Unlike in the classical SCS setting,
the presence of reverse complements introduces an additional selection constraint.
This constraint prevents a direct formulation
as a standard cycle cover problem
and precludes a straightforward reduction
to weighted bipartite matching.

More generally, consider a variant of the cycle cover problem
in which the vertex set is partitioned into several arbitrarily given disjoint subsets,
and one is required to select exactly one vertex from each subset.
The minimum-weight cycle cover problem under this constraint
is known to be NP-hard~\cite{noon1988generalized}.
Consequently, our constrained cycle cover problem
cannot be reduced to this more general variant
in a straightforward manner.

Although closely related cycle cover formulations have been studied
in the context of reverse complements, they typically permit self-loops,
which enables simple greedy algorithms~\cite{JIANG1992195, YamanoShibuya2026SCSRC}.
In contrast, excluding self-loops is essential in our setting
and makes these approaches inapplicable.

We overcome these difficulties by introducing a novel gadget construction
that reduces the resulting constrained minimum-weight cycle cover problem
to a weighted perfect matching problem in a general graph.

\begin{theorem} \label{thm: 2.67 approximation}
The SCS-RC problem admits an $\frac{8}{3}$-approximation algorithm.
\end{theorem}
On the hardness side, we show that approximation hardness for SCS transfers to SCS-RC.
\begin{restatable}{theorem}{SCStoSCSRCreduction}
\label{thm: SCS -> SCS-RC reduction}
If SCS-RC can be approximated within a factor of $\beta$, then SCS can also be approximated within a factor of $\beta$.
\end{restatable}
We further show that this hardness result persists even when the alphabet size is restricted to four, corresponding to the DNA alphabet
$\Sigma_{\mathrm{DNA}} = \{\texttt{A}, \texttt{C}, \texttt{G}, \texttt{T}\}$, where (\texttt{A}, \texttt{T}) and (\texttt{C}, \texttt{G}) are complementary pairs.
\begin{restatable}{theorem}{ToQuaternaryReduction}
\label{thm: arbitrary -> quaternary alphabet reduction}
If SCS-RC can be approximated within a factor of $\beta$ on instances over the DNA alphabet, then SCS-RC can also be approximated within a factor of $\beta$ on instances over an arbitrary alphabet.    
\end{restatable} 
By combining \cref{thm: SCS -> SCS-RC reduction,thm: arbitrary -> quaternary alphabet reduction} with the known hardness result that it is NP-hard to approximate SCS within a factor better than $\frac{333}{332}$ \cite{333/332.inapprox.CATS2013}, we obtain \cref{corollary: hardness result}.

\begin{corollary} \label{corollary: hardness result}
It is NP-hard to approximate SCS-RC within a factor better than $\frac{333}{332}$. This holds even for the DNA alphabet.
\end{corollary}

\section{Preliminaries}
Throughout this paper, unless stated otherwise,
we consider strings over a finite alphabet $\Sigma$
equipped with an involutive complement mapping
$cm \colon \Sigma \to \Sigma$,
that is,
$cm(cm(x)) = x$ for every $x \in \Sigma$.
A prominent example is the DNA alphabet
$\Sigma_{\mathrm{DNA}} = \{\texttt{A}, \texttt{C}, \texttt{G}, \texttt{T}\}$,
where the complement mapping is given by
$cm(\texttt{A}) = \texttt{T}$ and $cm(\texttt{C}) = \texttt{G}$,
reflecting the complementarity of DNA nucleotides.

For a string $s = s_1 s_2 \cdots s_n \in \Sigma^n$,
its \emph{reverse complement}, denoted by $\mathrm{rc}(s)$, is defined as
\[
\mathrm{rc}(s) = cm(s_n) cm(s_{n-1}) \cdots cm(s_1).
\]
This general setting allows us to capture not only
the standard reverse-complement operation on DNA,
but also the special case where
$cm(x) = x$ for all $x \in \Sigma$,
in which case $\mathrm{rc}(s)$ coincides with the usual reversal of $s$.

For a set of strings $T$, we define
$\mathrm{rc}(T) = \{\mathrm{rc}(t) \mid t \in T\}$.
Without loss of generality, we assume that the input set $S$ satisfies that
$S \cup \mathrm{rc}(S)$ is \emph{substring-free}, that is,
no string in $S \cup \mathrm{rc}(S)$ is a substring of another string in the same set. For a string $x$, we write $x' \in \{x, \mathrm{rc}(x)\}$ to represent
$x$ or its reverse complement $\mathrm{rc}(x)$.
The choice of $x'$ depends on the particular solution under consideration.
Let us formally define the
\emph{Shortest Common Superstring with Reverse Complements (SCS-RC)} problem.

\begin{definition}[Shortest Common Superstring with Reverse Complements (SCS-RC)]
\leavevmode\\
\textbf{Input:} A set of strings $S = \{s_1, \dots, s_m\}$ over $\Sigma$ such that
$S \cup \mathrm{rc}(S)$ is substring-free.\\
\textbf{Output:} A shortest string $s$ such that for every $s_i \in S$,
either $s_i$ or its reverse complement $\mathrm{rc}(s_i)$ appears as a substring of $s$.
\end{definition}

\subsection{The Distance Graph and Cycle Cover} \label{sec: distance graph and cycle cover}
We follow the standard notions of \emph{overlap}, \emph{prefix}, and \emph{distance}
as introduced in
\cite{Blum.et.al}.
For two strings $s$ and $t$, let $y$ be the longest string such that
$s = xy$ and $t = yz$ for some non-empty strings $x$ and $z$.
The length of $y$ is called the \emph{overlap} of $s$ and $t$ and is denoted by
$\op{ov}(s,t)$.
The string $x$ is called the \emph{prefix} of $s$ with respect to $t$ and is denoted by
$\op{pref}(s,t)$.
The length of $x$ is called the \emph{distance} from $s$ to $t$ and is denoted by
$\op{dist}(s,t)$.

For an ordered sequence of strings $x_1, \dots, x_r$,
we define
\[
\langle x_1, \dots, x_r \rangle
= \op{pref}(x_1, x_2)\op{pref}(x_2, x_3)
  \cdots
  \op{pref}(x_{r-1},x_r)x_r.
\]
Since $S \cup \op{rc}(S)$ is substring-free,
there exists an optimal solution to the SCS-RC problem of the form
$s = \langle s_{q_1}', \dots, s_{q_m}' \rangle$,
where $(q_1, \dots, q_m)$ is a permutation of $\{1, \dots, m\}$,
as observed in \cite{Blum.et.al} and extended to the reverse-complement setting in \cite{JIANG1992195}.
Let $\op{OPT}(S)$ denote the length of an optimal solution.

We consider the \emph{distance graph} induced by the input set of strings $S$
and denote it by $G_{\op{dist}}(S)$.
Formally, $G_{\op{dist}}(S) = (V, E, w)$ is a directed weighted graph defined as follows.
The vertex set is $V = S \cup \mathrm{rc}(S)$, and the edge set is
\[
E = (V \times V)
    \setminus 
    \bigcup_{x \in V} \{(x, x), (x, \op{rc}(x))\}.
\]
Each edge $(x,y) \in E$ is assigned weight $w(x,y) = \op{dist}(x,y)$.

Thus, $G_{\op{dist}}(S)$ contains neither self-loops nor edges between a string
and its reverse complement.
This definition differs from existing distance graph constructions
that take reverse complements into account
\cite{JIANG1992195, YamanoShibuya2026SCSRC},
in which self-loops are allowed.
By replacing the distance function with the overlap function,
we obtain the \emph{overlap graph} $G_{\op{ov}}(S)$.

For a cycle $c$ in a distance graph, let $w(c)$ denote the weight of the cycle, defined as the sum of the weights of its edges.

Given an instance $S$, consider the distance graph $G_{\op{dist}}(S)$.
Let $\op{GTSP}(G_{\op{dist}}(S))$ denote the minimum weight of a directed cycle that contains exactly one vertex from each pair $\{s_i, \mathrm{rc}(s_i)\}$ for $1 \le i \le m$.
This can be viewed as a special case of the generalized traveling salesman problem (GTSP) \cite{noon1988generalized}, where the clusters are $\{\{s_i, \mathrm{rc}(s_i)\} \mid 1 \le i \le m\}$.
Then,
\[
\op{GTSP}(G_{\op{dist}}(S)) \le \op{OPT}(S),
\]
and thus provides a lower bound for the SCS-RC problem.
Since GTSP is NP-hard \cite{noon1988generalized}, we consider a relaxation of this formulation, namely the \emph{cycle cover} problem.

A \emph{cycle cover} of $G_{\op{dist}}(S)$ is a collection of vertex-disjoint directed cycles that contains exactly one vertex from each pair $\{s_i, \mathrm{rc}(s_i)\}$ for $1 \le i \le m$. 
This is a constrained variant of the classical cycle cover problem. 
We denote by $\op{CYC}(G_{\op{dist}}(S))$ a minimum-weight cycle cover, and by $w(\op{CYC}(G_{\op{dist}}(S)))$ its total weight. 
Clearly,
\[
w(\op{CYC}(G_{\op{dist}}(S))) \le \op{GTSP}(G_{\op{dist}}(S)) \le \op{OPT}(S).
\]

A similar notion applies to the overlap graph $G_{\op{ov}}(S)$.
In this case, one considers a maximum-weight cycle cover, which
corresponds exactly to the minimum-weight cycle cover in the
distance graph $G_{\op{dist}}(S)$. This correspondence follows from
the relation
$\op{ov}(x, y) = |x| - \op{dist}(x, y),$
which holds for any edge from $x$ to $y$ in the graph.

\subsection{Periodicity of Strings}
We follow the definitions in \cite{Breslaure.1997.OverlapRotationLemma}.
A string $x$ is a \emph{factor} of a finite string $s$ if
$s = x^{i} y$ for some integer $i \ge 1$ and some (possibly empty) prefix $y$ of $x$.
Let $\op{factor}(s)$ be the shortest such $x$, and define
$\op{period}(s) = |\op{factor}(s)|$.
A \emph{periodic semi-infinite} string $s$ is an infinite string satisfying
$s = x s$ for some nonempty string $x$.
The shortest such $x$ is denoted by $\op{factor}(s)$, and
$\op{period}(s) = |\op{factor}(s)|$.
Let $s$ and $t$ be strings, each of which is either finite or periodic semi-infinite.
We say that $s$ and $t$ are \emph{equivalent} if $\op{factor}(s)$ and $\op{factor}(t)$ are cyclic shifts of each other; that is, there exist strings $a$ and $b$ such that
$\op{factor}(s) = ab$ and $\op{factor}(t) = ba$.
Otherwise, $s$ and $t$ are said to be \emph{inequivalent}.
This defines an equivalence relation on such strings.
In particular, equivalent strings have the same period.

\section{Useful Lemmas from Previous Work}
\cref{lm: cycle weight equals period} relates the periodicity of overlapped strings to cycles in the distance graph.
This result was proved in \cite{Blum.et.al} and restated in \cite{Breslaure.1997.OverlapRotationLemma}.

\begin{lemma}[\cite{Blum.et.al}] \label{lm: cycle weight equals period}
Let
$c = (s_{i_1}', \dots, s_{i_r}', s_{i_1}')$
be a cycle in $\op{CYC}(G_{\op{dist}}(S))$.
Then
\[
w(c)
=
\op{dist}(s_{i_1}', s_{i_2}') + \cdots + \op{dist}(s_{i_{r-1}}', s_{i_r}') + \op{dist}(s_{i_r}', s_{i_1}')
=
\op{period}(\langle s_{i_1}', \dots, s_{i_r}' \rangle).
\]
Moreover, the strings
$\langle s_{i_1}', \dots, s_{i_r}' \rangle,
\langle s_{i_2}', \dots, s_{i_r}', s_{i_1}' \rangle,
\dots,
\langle s_{i_r}', s_{i_1}', \dots, s_{i_{r-1}}' \rangle$
are all equivalent.
\end{lemma}
The inequivalence of strings extracted from distinct cycles
of a minimum-weight cycle cover
was originally shown in \cite{Blum.et.al},
and later extended to the reverse-complement setting
in \cite{JIANG1992195,YamanoShibuya2026SCSRC}.

\begin{lemma}[\cite{Blum.et.al,JIANG1992195}] \label{lm: distinct cycles are inequivalent}
Let
$c = (s_{i_1}', \dots, s_{i_r}', s_{i_1}')$
and
$d = (s_{j_1}', \dots, s_{j_k}', s_{j_1}')$
be two distinct cycles in
$\op{CYC}(G_{\op{dist}}(S))$.
Define
$e = \langle s_{i_1}', \dots, s_{i_r}' \rangle$
and
$f = \langle s_{j_1}', \dots, s_{j_k}' \rangle$.
Then the four pairs
$
(e, f),\ (e, \op{rc}(f)),\ (\op{rc}(e), f),\ (\op{rc}(e), \op{rc}(f))
$
are all inequivalent.
\end{lemma}

Given a semi-infinite string $\alpha = a_1a_2\cdots$, we denote by $\alpha[k] = a_k a_{k+1} \cdots$ its rotation starting at position $k$. 
Breslauer et al.\ \cite{Breslaure.1997.OverlapRotationLemma} proved the following overlap rotation lemma.

\begin{lemma}[Overlap Rotation Lemma \cite{Breslaure.1997.OverlapRotationLemma}] \label{lm: overlap rotation lemma}
Let $\alpha$ be a periodic semi-infinite string. 
Then there exists an integer $k$ such that, for any finite string $s$ that is inequivalent to $\alpha$,
\[
\text{if}~\op{period}(s) \le \op{period}(\alpha), \text{then}~\op{ov}(s, \alpha[k]) \le \frac{2}{3}(\op{period}(s) + \op{period}(\alpha)).
\]
\end{lemma}

We refer to a rotation $\alpha[k]$ obtained from
\cref{lm: overlap rotation lemma}
as the \emph{critical rotation}.
Lemma~5.1 of \cite{Breslaure.1997.OverlapRotationLemma}
was later restated in the reverse-complement setting
as Lemma~14 in \cite{YamanoShibuya2026SCSRC}.
The inequalities in property~(4) follow from the fact that, for distinct cycles $c$ and $d$,
the corresponding strings $t_c$ and $t_d$ are inequivalent from property~(3) and \cref{lm: distinct cycles are inequivalent}, 
while $\op{period}(t_c) = w(c)$ and $\op{period}(t_d) = w(d)$ from property~(3) and \cref{lm: cycle weight equals period}.

\begin{lemma}[\cite{Breslaure.1997.OverlapRotationLemma,YamanoShibuya2026SCSRC}]
\label{lm: extract good rotation}
Let
$c = (s_{i_1}', \dots, s_{i_r}', s_{i_1}')$
be a cycle in $\op{CYC}(G_{\op{dist}}(S))$.
Then there exist a string $t_c$ and an index $j$ such that the following properties hold:
\begin{bracketenumerate}
    \item
    The string
    $\langle s_{i_{j+1}}', \dots, s_{i_r}', s_{i_1}', \dots, s_{i_j}' \rangle$
    is a suffix of $t_c$.
    \item
    The string $t_c$ is a substring of
    $\langle s_{i_j}', \dots, s_{i_r}', s_{i_1}', \dots, s_{i_j}' \rangle$.
    \item
    The string $t_c$ is equivalent to
    $\langle s_{i_{j+1}}', \dots, s_{i_r}', s_{i_1}', \dots, s_{i_j}' \rangle$.
    \item
    The semi-infinite string $\op{factor}(t_c)^\infty$
    is the critical rotation of
    $\op{factor}(\langle s_{i_1}', \dots, s_{i_r}' \rangle)^\infty$.
    Moreover, let $t_d$ be the string obtained from this lemma corresponding to a distinct cycle $d \in \op{CYC}(G_{\op{dist}}(S))$, where $w(d) \le w(c)$. Then the following inequations hold:
    \[
        \op{ov}(t_d, t_c)
        \le \frac{2}{3} (w(d) + w(c)),
        \quad
        \op{ov}(\op{rc}(t_d), t_c)
        \le \frac{2}{3} (w(d) + w(c)).
    \]
\end{bracketenumerate}
Furthermore, the string $t_c$ can be computed in time linear in $w(c)$.
\end{lemma}

Lemma~2.6 of \cite{KAPLAN200513.greedy3.5n} was adapted
to the reverse-complement setting in Lemma~17 of \cite{YamanoShibuya2026SCSRC}.

\begin{lemma}[\cite{KAPLAN200513.greedy3.5n,YamanoShibuya2026SCSRC}]
\label{lm: bound for OPT of tc}
Let
$T = \{t_c \mid c \in \op{CYC}(G_{\op{dist}}(S))\}$,
where each string $t_c$ is obtained as in \cref{lm: extract good rotation}.
Then
\[
    \op{OPT}(T)
    \le
    \op{OPT}(S) + w(\op{CYC}(G_{\op{dist}}(S))).
\]
\end{lemma}

\section{Approximation Algorithm} \label{sec: Approximation Algorithm}
We begin by presenting an overview in
\cref{alg: 2.67-approx algorithm overview},
which follows the framework of \cite{Breslaure.1997.OverlapRotationLemma}.

\subsection{Overview of the Entire \texorpdfstring{$\frac{8}{3}$}{8/3}-Approximation Algorithm}
In Step~2 of \cref{alg: 2.67-approx algorithm overview}, the string $t_j$
contains, for each vertex of the cycle $c_j$, the string corresponding to
that vertex, by Property~(1) of \cref{lm: extract good rotation}.
In Step~4, the resulting string is
$\langle t_{i_p}', \dots, t_{i_r}', t_{i_1}', \dots, t_{i_{p-1}}' \rangle$
if $t_{i_p}' = t_{i_p}$, and
$\langle t_{i_{p+1}}', \dots, t_{i_r}', t_{i_1}', \dots, t_{i_p}' \rangle$
otherwise.
In both cases, the resulting string contains, for each vertex of the cycle
$d$, the string corresponding to that vertex.
Together, these two observations ensure that
\cref{alg: 2.67-approx algorithm overview} outputs a valid approximate
solution of SCS-RC for the instance $S$.

In Step~4, we break each cycle with particular care.
Specifically, we remove an edge whose overlap is bounded as guaranteed by Property~(4) of
\cref{lm: extract good rotation},
thereby ensuring that the increase in length incurred by breaking the cycle
is properly controlled.
Observe that the length of the resulting string in Step~4 is
$w(d) + \op{ov}(t_{i_{p-1}}', t_{i_p}')$ when $t_{i_p}' = t_{i_p}$, and
$w(d) + \op{ov}(t_{i_p}', t_{i_{p+1}}')$ otherwise.

Note that self-loops are not permitted in this step.
Accordingly, self-loops are excluded in our definition of distance graphs, and hence do not appear in the minimum-weight cycle covers.
The algorithmic details of computing such cycle covers are deferred to
\cref{sec: Compute cycle cover}.

In the remainder of this subsection, we assume that the minimum-weight
cycle cover can be computed in polynomial time, and focus on the approximation analysis.
Specifically, we show that
\cref{alg: 2.67-approx algorithm overview}
achieves an approximation ratio of $\frac{8}{3}$.

\begin{theorem} \label{thm: Alg1 yields a 2.67 approximation}
\cref{alg: 2.67-approx algorithm overview} yields an
$\frac{8}{3}$-approximation for the SCS-RC problem.
\end{theorem}

\begin{proof}
Let $\op{ALG}(S)$ denote the length of the string produced by
\cref{alg: 2.67-approx algorithm overview}.
By Step~4 of the algorithm, the length $\op{ALG}(S)$ can be expressed as
\[
\op{ALG}(S)
=
w(\op{CYC}(G_{\op{dist}}(T)))
+
\sum_{d \in \op{CYC}(G_{\op{dist}}(T))} OV_d,
\]
where $OV_d$ denotes the overlap lost when breaking the cycle $d$.
Combining \cref{lm: bound for OPT of tc} with the inequality
$w(\op{CYC}(G_{\op{dist}}(T))) \le \op{OPT}(T)$, we obtain
\[
w(\op{CYC}(G_{\op{dist}}(T)))
\le
\op{OPT}(S) + w(\op{CYC}(G_{\op{dist}}(S))).
\]
We next bound the term $OV_d$ for each cycle $d$.
Consider first the case where $t_{i_p}' = t_{i_p}$ in Step~4.
For convenience, we define $t_{i_0}' = t_{i_r}'$ and $c_{i_0} = c_{i_r}$.
In this case, the overlap lost is $\op{ov}(t_{i_{p-1}}', t_{i_p}') = \op{ov}(t_{i_{p-1}}', t_{i_p})$.
By \cref{lm: extract good rotation}, we have
\[
\op{ov}(t_{i_{p-1}}', t_{i_p})
\le
\frac{2}{3}\bigl(w(c_{i_{p-1}}) + w(c_{i_{p}})\bigr)
\le
\frac{2}{3} \sum_{\{j \mid t_j' \in d\}} w(c_{j}).
\]
Now consider the case where $t_{i_p}' = \op{rc}(t_{i_p})$.
We define $t_{i_{r+1}}' = t_{i_1}'$ and
$c_{i_{r+1}} = c_{i_1}$.
The overlap lost in this case is
$\op{ov}(t_{i_p}', t_{i_{p+1}}') = \op{ov}(\op{rc}(t_{i_p}), t_{i_{p+1}}')$.
Again by \cref{lm: extract good rotation}, we obtain
\[
\op{ov}(\op{rc}(t_{i_p}), t_{i_{p+1}}')
=
\op{ov}(\op{rc}(t_{i_{p+1}}'), t_{i_p})
\le
\frac{2}{3}\bigl(w(c_{i_{p + 1}}) + w(c_{i_{p}})\bigr)
\le
\frac{2}{3} \sum_{\{j \mid t_j' \in d\}} w(c_{j}).
\]
In both cases, the overlap lost for cycle $d$ is bounded by
$\frac{2}{3} \sum_{\{j \mid t_j' \in d\}} w(c_{j})$.
Summing over all cycles yields
\[
\sum_{d \in \op{CYC}(G_{\op{dist}}(T))} OV_d
\le
\frac{2}{3}
\sum_{d \in \op{CYC}(G_{\op{dist}}(T))}
\sum_{\{j \mid t_j' \in d\}} w(c_{j})
=
\frac{2}{3} w(\op{CYC}(G_{\op{dist}}(S))).
\]
Combining the above bounds, we conclude that
\[
\op{ALG}(S)
\le
\op{OPT}(S) + \frac{5}{3} w(\op{CYC}(G_{\op{dist}}(S)))
\le
\frac{8}{3} \op{OPT}(S),
\]
which completes the proof.
\end{proof}

\begin{algorithm}[t]
\caption{$\frac{8}{3}$-approximation algorithm for SCS-RC}
\label{alg: 2.67-approx algorithm overview}
\begin{enumerate}
    \item Construct the distance graph $G_{\op{dist}}(S)$ from the input string
    set $S$, and compute a minimum-weight cycle cover
    $\op{CYC}(G_{\op{dist}}(S)) = \{c_1, \dots, c_k\}$.

    \item For each $j \in \{1,\dots,k\}$, let $t_j$ be the string obtained from
    the cycle $c_j$ as described in \cref{lm: extract good rotation}, and let
    $T = \{t_1, \dots, t_k\}$.
    Construct the distance graph $G_{\op{dist}}(T)$.

    \item Compute a minimum-weight cycle cover
    $\op{CYC}(G_{\op{dist}}(T))$.

    \item For each cycle
    $d = (t_{i_1}', \dots, t_{i_r}', t_{i_1}')$
    in $\op{CYC}(G_{\op{dist}}(T))$,
    let $p \in \{1,\dots,r\}$ be an index such that
    $t_{i_p}'$ has maximum period among the vertices of $d$.
    If $t_{i_p}' = t_{i_p}$, break the cycle by deleting the edge incoming to
    $t_{i_p}'$; otherwise, delete the edge outgoing from $t_{i_p}'$.
    This produces a superstring containing all strings corresponding to the
    vertices of $d$.

    \item Concatenate the resulting strings in an arbitrary order.
\end{enumerate}
\end{algorithm}

\subsection{Computing the Constrained Optimal Cycle Cover}
\label{sec: Compute cycle cover}

We consider computing a maximum-weight cycle cover in the overlap graph,
rather than a minimum-weight cycle cover in the distance graph.
As discussed in \cref{sec: distance graph and cycle cover},
these two formulations are equivalent.
However, the overlap graph admits the symmetry
\[
    \op{ov}(s,t) = \op{ov}(\op{rc}(t), \op{rc}(s))
\]
for any strings $s$ and $t$,
which simplifies the presentation of our algorithm.
In the remainder of this subsection,
we fix an input set of strings $S$
and consider the problem of computing
a maximum-weight cycle cover in the overlap graph $G_{\op{ov}}(S)$.

\paragraph*{Overview of the cycle cover construction}
Our algorithm does not compute a cycle cover satisfying all constraints
in a single step; instead, it proceeds in two stages.

In the first stage, we impose only the constraint that,
for any pair of strings $s$ and $t$,
the edges $s \to t$ and $\op{rc}(t) \to \op{rc}(s)$
cannot be selected simultaneously.
This constraint is enforced via a gadget construction,
under which we compute a maximum-weight perfect matching.
From the resulting matching, we derive a second, symmetric matching
by replacing each edge $s \to t$ with its reverse-complement counterpart
$\op{rc}(t) \to \op{rc}(s)$.
The union of these two matchings
admits the extraction of a maximum-weight cycle cover
satisfying the constraints,
together with its reverse-complement image.
The overall procedure is illustrated in \cref{fig: Alg All}.

\paragraph*{Construction of the auxiliary graph $G'$ with gadgets}
For technical convenience, we work with oriented copies of strings.
For each string $s \in S$, we introduce two formal symbols
$s^+$ and $s^-$, representing $s$ and its reverse complement, respectively.
Let
\[
\tilde{S} = \bigcup_{s \in S} \{ s^+, s^- \}.
\]
Even if $s = \op{rc}(s)$, the two elements $s^+$ and $s^-$
are treated as distinct vertices.
This convention allows us to uniformly model the problem as selecting exactly one representative from each pair $\{s^+, s^-\}$, regardless of whether $s = \op{rc}(s)$ or not.

For $u = s^+$ and $v = s^-$, we define $\bar{u} = v$ and $\bar{v} = u$.
More generally, for any $u \in \tilde{S}$, we denote by $\bar{u}$
the oriented copy corresponding to the reverse complement.
We also associate each $u \in \tilde{S}$ with a string
$\op{str}(u)$, defined by
\[
\op{str}(s^+) = s
\quad\text{and}\quad
\op{str}(s^-) = \op{rc}(s).
\]

We recall the standard reduction from cycle covers to perfect matchings.
Given a directed graph, a cycle cover can be computed by reducing the problem to a perfect matching in a bipartite graph, where each vertex is split into an outgoing and an incoming copy.

As in this reduction,
each vertex $u \in \tilde{S}$ is split into two vertices
$u_{\op{out}}$ and $u_{\op{in}}$,
representing outgoing and incoming edges, respectively.
However, unlike the standard setting,
we treat the two vertices $\{u, \bar{u}\}$ as a single cluster.

Formally, we first define a bipartite graph
$\tilde{G} = (V_{\op{out}}, V_{\op{in}}, \tilde{E}, w)$
constructed from the overlap graph $G_{\op{ov}}(S)$.
The vertex sets are defined as
\[
V_{\op{out}} = \{ u_{\op{out}} \mid u \in \tilde{S} \}
\quad\text{and}\quad
V_{\op{in}} = \{ u_{\op{in}} \mid u \in \tilde{S} \}.
\]

The edge set $\tilde{E}$ is defined by
\[
\tilde{E} =
(V_{\op{out}} \times V_{\op{in}})
\setminus
\bigcup_{u \in \tilde{S}}
\bigl\{ (u_{\op{out}}, u_{\op{in}}),
       (\bar{u}_{\op{out}}, \bar{u}_{\op{in}}) \bigr\}.
\]
For each edge $(u_{\op{out}}, v_{\op{in}}) \in \tilde{E}$,
the weight is defined as
$w(u_{\op{out}}, v_{\op{in}}) =
\op{ov}(\op{str}(u), \op{str}(v))$.

\begin{figure}[t]
  \centering
  \begin{minipage}[t]{0.475\linewidth}
    \centering
    \includegraphics[width=\linewidth]{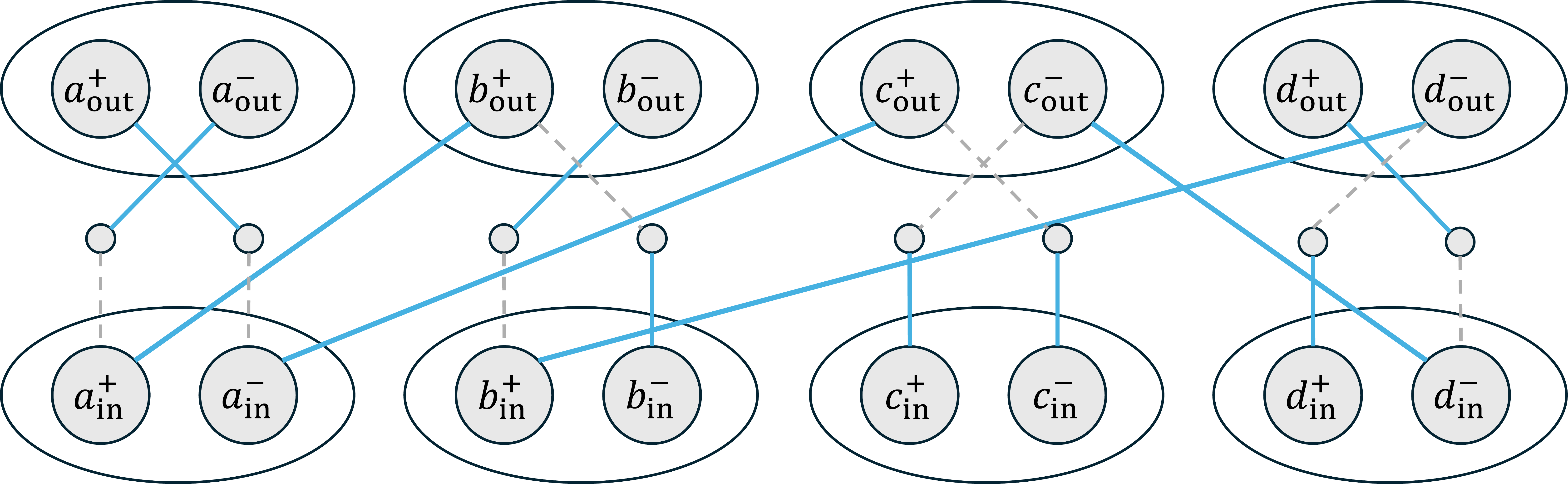}
    \subcaption{Compute a maximum-weight perfect matching in $G'$.
The small vertices in the center correspond to auxiliary vertices.
The matching is shown in light blue.
Dashed gray edges represent edges incident to auxiliary vertices
that are not included in the matching.}
    \label{fig: alg1}
  \end{minipage}\hfill
  \begin{minipage}[t]{0.475\linewidth}
    \centering
    \includegraphics[width=\linewidth]{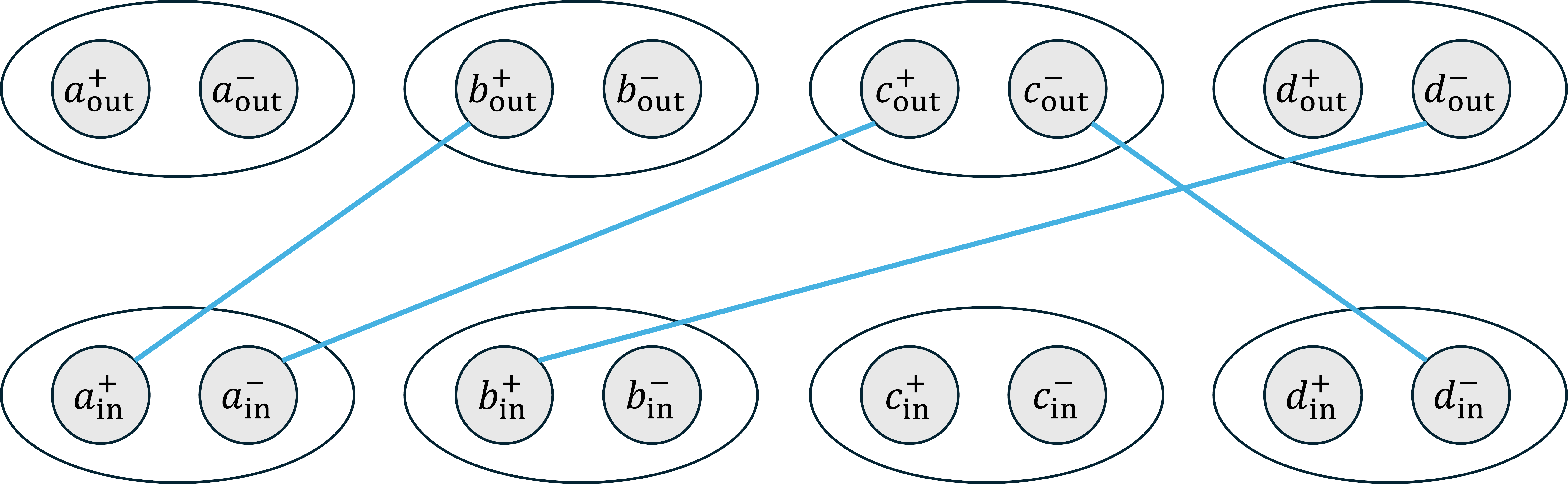}
    \subcaption{Remove all edges incident to auxiliary vertices to obtain the edge set $F$.}
    \label{fig: alg2}
  \end{minipage}

  \medskip

  \begin{minipage}[t]{0.475\linewidth}
    \centering
    \includegraphics[width=\linewidth]{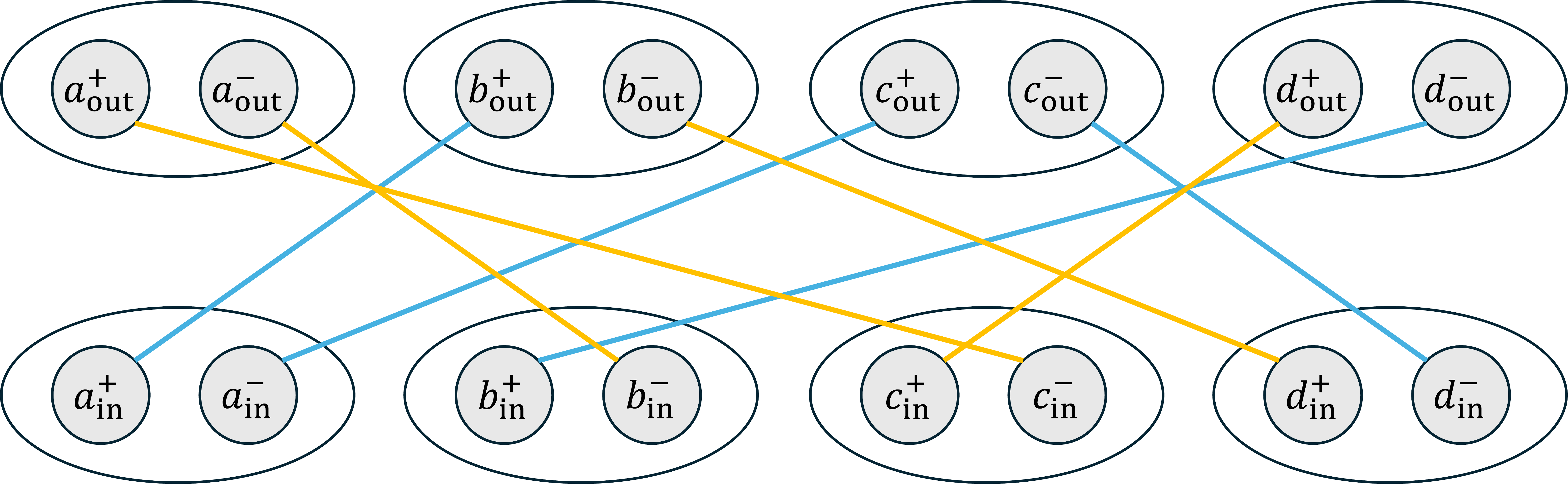}
    \subcaption{Compute the reverse-complement image $\bar{F}$, shown in yellow.}
    \label{fig: alg3}
  \end{minipage}\hfill
  \begin{minipage}[t]{0.475\linewidth}
    \centering
    \includegraphics[width=\linewidth]{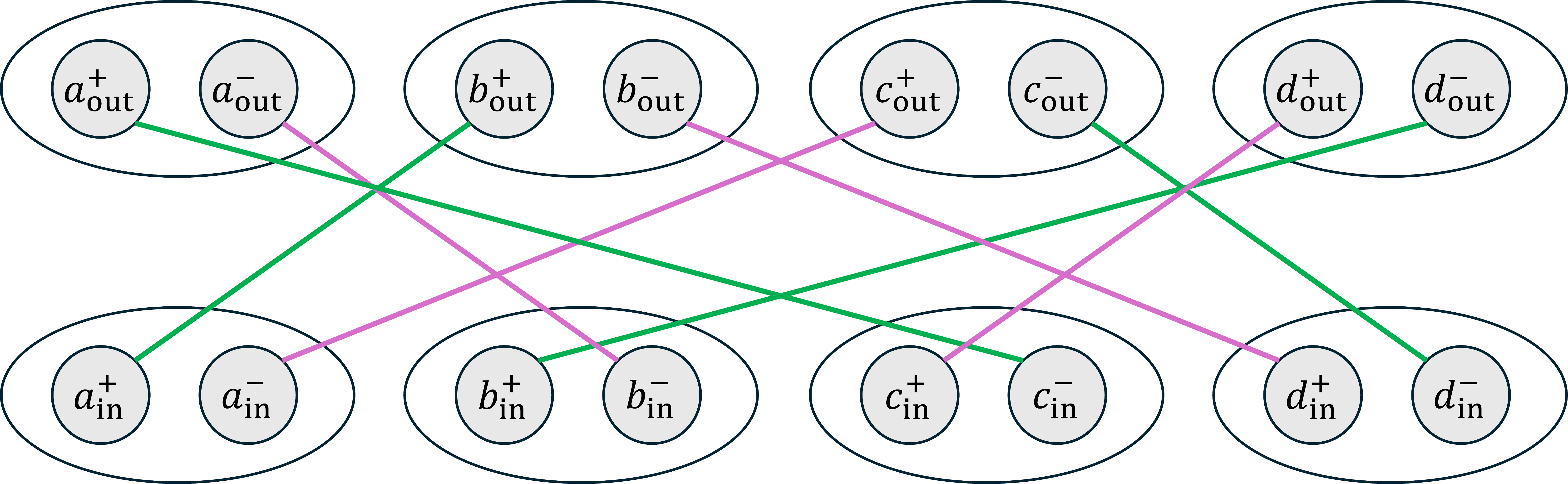}
    \subcaption{Extract two matchings $M$ and $\bar{M}$ corresponding to the resulting maximum-weight cycle cover.
Matching $M$, shown in green, corresponds to the cycle
$(a, \op{rc}(c), \op{rc}(d), b, a)$, whereas matching $\bar{M}$,
shown in purple, corresponds to the cycle
$(\op{rc}(a), \op{rc}(b), d, c, \op{rc}(a))$.}
    \label{fig: alg4}
  \end{minipage}

  \caption{An example illustrating the computation of a maximum-weight
cycle cover in an overlap graph induced by the string set
$S = \{a, b, c, d\}$.}
  \label{fig: Alg All}
\end{figure}

A matching in $\tilde{G}$ that corresponds to a valid cycle cover
of the overlap graph $G_{\op{ov}}(S)$
must satisfy the following condition.
For each $u \in \tilde{S}$, exactly one of the following two cases holds:
\begin{romanenumerate}
    \item $u_{\op{out}}$ and $u_{\op{in}}$ are matched, and $\bar{u}_{\op{out}}$ and $\bar{u}_{\op{in}}$ are unmatched;
    \item $\bar{u}_{\op{out}}$ and $\bar{u}_{\op{in}}$ are matched, and $u_{\op{out}}$ and $u_{\op{in}}$ are unmatched.
\end{romanenumerate}
Case (i) corresponds to selecting the string $\op{str}(u)$,
whereas case (ii) corresponds to selecting its reverse complement $\op{str}(\bar{u})$.
Enforcing this constraint directly within a matching formulation,
however, turns out to be technically challenging.

Instead, we impose a weaker constraint:
for each $u \in \tilde{S}$,
the vertices $u_{\op{out}}$ and $\bar{u}_{\op{in}}$
cannot be matched simultaneously.
To enforce this condition, we construct a graph $G'$
by augmenting $\tilde{G}$ with the following gadget.
For each $u \in \tilde{S}$,
we introduce an auxiliary vertex $x_u$
and add two zero-weight edges
$(u_{\op{out}}, x_u)$ and $(x_u, \bar{u}_{\op{in}})$.

\paragraph*{Obtaining the edge sets $F$ and $\bar{F}$ from a maximum-weight perfect matching in $G'$}
The graph $G'$ may no longer be bipartite due to the introduction of auxiliary gadgets.
Accordingly, we compute a maximum-weight perfect matching in $G'$
using an algorithm for general graphs.
From the obtained matching, we discard all edges incident to auxiliary vertices
and denote the remaining edge set by $F \subseteq \tilde{E}$.
An example of this procedure is shown in \cref{fig: alg1,fig: alg2}.
Since the matching is computed under a relaxation of the constraints
required for a valid cycle cover,
the total weight of $F$ is at least that of a maximum-weight cycle cover
of $G_{\op{ov}}(S)$.

The constraints enforced by the gadgets ensure that, for each
$u \in \tilde{S}$, edges outgoing from $u_{\op{out}}$ and edges incoming to
$\bar{u}_{\op{in}}$ are not selected simultaneously in $F$.
This implies that, for any distinct $u, v \in \tilde{S}$, the two edges
$(u_{\op{out}}, v_{\op{in}})$ and
$(\bar{v}_{\op{out}}, \bar{u}_{\op{in}})$
cannot be selected simultaneously.
Consequently, we can construct the reverse-complement image of $F$
by replacing each edge
$(u_{\op{out}}, v_{\op{in}}) \in F$
with the edge
$(\bar{v}_{\op{out}}, \bar{u}_{\op{in}})$.
We denote the resulting edge set by $\bar{F}$.
An example of this construction is shown in \cref{fig: alg3}.
By the symmetry of the overlap weights, namely
$\op{ov}(\op{str}(u), \op{str}(v)) = \op{ov}(\op{str}(\bar{v}), \op{str}(\bar{u}))$,
the total weight of $\bar{F}$ is equal to that of $F$.
Observe that the union $F \cup \bar{F}$ is a perfect matching of $\tilde{G}$.

\paragraph*{Extracting two optimal cycle covers from $F \cup \bar{F}$}
We prove that there exists a matching $M$ in $\tilde{G}$ corresponding to a maximum-weight cycle cover of the overlap graph $G_{\op{ov}}(S)$ such that $F \cup \bar{F} = M \cup \bar{M}$, where $\bar{M}$ denotes the reverse-complement image of $M$. An example illustrating the extraction of $M$ and $\bar{M}$ is shown in \cref{fig: alg4}. Since $F \cup \bar{F}$ is a perfect matching of $\tilde{G}$, it induces a collection of vertex-disjoint directed cycles covering all vertices of $\tilde{S}$. Moreover, the total weight of $F \cup \bar{F}$ is an upper bound on the total weight of $M \cup \bar{M}$, since $F$ is obtained from a relaxation of the constraints, as discussed above. Therefore, it suffices to show that each such cycle is valid with respect to the definition of a cycle cover in the original overlap graph $G_{\op{ov}}(S)$.

\begin{lemma} \label{lm: Valid Extraction}
No cycle induced by $F \cup \bar{F}$ contains both
a vertex $u \in \tilde{S}$ and its reverse-complement counterpart $\bar{u}$.
\end{lemma}

\begin{proof}
Suppose, for the sake of contradiction, that there exists a directed cycle
\[
c = (v_1, \ldots, \bar{v}_1, \ldots, v_1)
\]
that contains both a vertex $v_1 \in \tilde{S}$ and its reverse-complement counterpart
$\bar{v}_1$.
Let $v_i \to v_{i+1}$ be an arbitrary directed edge of $c$.

By construction, for every edge $(e_{\op{out}}, f_{\op{in}}) \in F \cup \bar{F}$,
its reverse-complement image $(\bar{f}_{\op{out}}, \bar{e}_{\op{in}})$
also belongs to $F \cup \bar{F}$.
Therefore, whenever the cycle $c$ contains the edge
$v_i \to v_{i+1}$ and also contains the vertex $\bar{v}_i$,
it must contain the edge
$\bar{v}_{i+1} \to \bar{v}_i$ as well.

Consider traversing the cycle $c$ simultaneously from $v_1$ forward
and from $\bar{v}_1$ backward.
Since both $v_1$ and $\bar{v}_1$ lie on the same directed cycle,
the two traversals must eventually meet.
At the meeting point, there exists an index $i$ such that
both $v_i$ and its reverse-complement $\bar{v}_i$
appear consecutively on $c$,
implying that $c$ contains a directed edge of the form
$v_i \to \bar{v}_i$.

However, such an edge cannot exist:
by definition of $\tilde{G}$, the edge set $\tilde{E}$
contains no edge connecting a vertex to its reverse complement.
This contradiction completes the proof.
\end{proof}

By \cref{lm: Valid Extraction}, for each cycle $c$ induced by $F \cup \bar{F}$,
its reverse-complement image $\bar{c}$ is also induced by $F \cup \bar{F}$ and forms a distinct vertex-disjoint cycle.
Hence, every induced cycle belongs to a unique pair $\{c,\bar{c}\}$, and from each pair, we can arbitrarily choose exactly one cycle.
The selected cycles form a collection of vertex-disjoint directed cycles
that covers exactly one vertex from each pair $\{u, \bar{u}\}$ for $u \in \tilde{S}$,
and hence yields a valid cycle cover of $G_{\op{ov}}(S)$.

To conclude this section, we have shown that a maximum-weight cycle cover
of the overlap graph $G_{\op{ov}}(S)$ can be computed in polynomial time.
Using the equivalence between the maximum-weight cycle cover of $G_{\op{ov}}(S)$
and the minimum-weight cycle cover of $G_{\op{dist}}(S)$,
this establishes \cref{thm: compute cycle cover}.

\begin{theorem} \label{thm: compute cycle cover}
Given a distance graph $G_{\op{dist}}(S)$,
the minimum-weight cycle cover \\
$\op{CYC}(G_{\op{dist}}(S))$ can be computed in polynomial time.
\end{theorem}

Combining \cref{thm: compute cycle cover}
with \cref{thm: Alg1 yields a 2.67 approximation},
we obtain a polynomial-time $\frac{8}{3}$-approximation algorithm for SCS-RC.
This completes the proof of our first main result,
\cref{thm: 2.67 approximation} from \cref{sec: Our Contributions}.

\section{Hardness Results}
In this section, we establish the hardness results stated in
\cref{thm: SCS -> SCS-RC reduction,thm: arbitrary -> quaternary alphabet reduction}
from \cref{sec: Our Contributions}.
We restate the theorems here for ease of reference.
\SCStoSCSRCreduction*
\ToQuaternaryReduction*
We first prove \cref{thm: SCS -> SCS-RC reduction} by presenting a reduction from SCS to SCS-RC that preserves the optimal solution length.
We then prove \cref{thm: arbitrary -> quaternary alphabet reduction}
by presenting a ratio-preserving reduction from SCS-RC over an arbitrary alphabet
to SCS-RC over the DNA alphabet.

\subsection{Proof of Theorem~\ref{thm: SCS -> SCS-RC reduction}}
Let $S = \{s_1, \dots, s_m\}$ be an arbitrary instance of SCS, and let
$A = \{a_1, \dots, a_k\}$ denote the alphabet consisting of all characters
appearing in the strings of $S$.
From this instance, we construct a corresponding instance of SCS-RC as follows.
Define a new alphabet
\[
B = \{a_1, \dots, a_k, a_{k+1}, \dots, a_{2k}\},
\]
and introduce a complement mapping
$cm : B \rightarrow B$ defined by
\[
cm(a_i) =
\begin{cases}
a_{i+k}, & \text{if $1 \le i \le k$,} \\
a_{i-k}, & \text{if $k+1 \le i \le 2k$.}
\end{cases}
\]
It is immediate that $cm$ is an involution on $B$.

Using the alphabet $B$ equipped with the complement mapping $cm$,
we consider the SCS-RC instance with the same input set $S$.
We show that any approximate solution to this SCS-RC instance
can be transformed into a solution that uses only strings from $S$,
without increasing its length.

Let $t$ be an arbitrary approximate solution to the constructed SCS-RC instance.
There exists a permutation $(i_1, \dots, i_m)$ of $\{1, \dots, m\}$ such that
$t$ can be written as
\[
t = \langle s'_{i_1}, \dots, s'_{i_m} \rangle.
\]

We define a set of intervals $I$ as follows:
\begin{align*}
I = \{ [l, r] \mid\;& 1 \le l \le r \le m,~\text{and}~(l = 1 \ \text{or}\ s'_{i_{l-1}} = s_{i_{l-1}}),~\text{and}~(r = m \ \text{or}\ s'_{i_{r+1}} = s_{i_{r+1}}), \\
&\text{and}~s'_{i_k} = \op{rc}(s_{i_k}) \ \text{for all } k \in \{l, \dots, r\} \}.
\end{align*}
That is, each interval $[l,r] \in I$ corresponds to a maximal substring
of reverse-complemented strings in $t$.
For each interval $[l, r] \in I$, the string $t$ contains the substring
\[
x = \langle \op{rc}(s_{i_l}), \dots, \op{rc}(s_{i_r}) \rangle.
\]
We replace this substring by its reverse complement
\[
\op{rc}(x) = \langle s_{i_r}, \dots, s_{i_l} \rangle.
\]

Since the complement mapping $cm$ maps symbols in $A$
to symbols in $B \setminus A$ and vice versa,
no overlap can occur between strings from $S$
and strings from $\op{rc}(S)$.
Therefore, this transformation does not decrease any existing overlap,
and hence does not increase the total length of the superstring.
By repeatedly applying this operation to all intervals in $I$,
we eventually obtain a solution that consists solely of strings from $S$,
and whose length is no greater than that of the original solution.

In particular, any optimal solution to the original SCS instance
is also an optimal solution to the constructed SCS-RC instance.
Moreover, any $\beta$-approximate solution to the constructed SCS-RC instance
can be transformed into a solution to the original SCS instance of the same length.
Hence, a $\beta$-approximation algorithm for SCS-RC
yields a $\beta$-approximation algorithm for SCS.

\subsection{Proof of Theorem \ref{thm: arbitrary -> quaternary alphabet reduction}}
Let $S = \{s_1, \dots, s_m\}$ be an arbitrary instance of SCS-RC, and let
$A = \{a_1, \dots, a_k\}$ be the alphabet consisting of all characters
appearing in the strings of $S \cup \op{rc}(S)$, equipped with a complement mapping
$cm \colon A \rightarrow A$.
Since $cm$ is an involution, for each $a \in A$, either $cm(a) = a$ or $cm(a) \ne a$ and $cm(cm(a)) = a$.
Without loss of generality, we may relabel the symbols in $A$ so that there
exists an index $j$ with $0 \le j \le \lfloor \frac{k}{2} \rfloor$ such that
\[
cm(a_i) = 
\begin{cases}
a_{i+j}~&\text{if } 1 \le i \le j, \\
a_{i-j}~&\text{if } j+1 \le i \le 2j, \\
a_i~&\text{if } 2j+1 \le i \le k.
\end{cases}
\]
Here, we define a morphism
$h \colon A^* \rightarrow \Sigma_{\mathrm{DNA}}^*$.
First, for each symbol $a_i \in A$, let
\[
h(a_i) =
\begin{cases}
\texttt{A}^i (\texttt{A}\texttt{G})^{(k+1-i-j)} \texttt{G}^i
& \text{if } 1 \le i \le j, \\[2mm]
\texttt{C}^{(i-j)} (\texttt{C}\texttt{T})^{(k+1-i)} \texttt{T}^{(i-j)}
& \text{if } j+1 \le i \le 2j, \\[2mm]
\texttt{A}^{(i-j)} (\texttt{A}\texttt{T})^{(k+1-i)} \texttt{T}^{(i-j)}
& \text{if } 2j+1 \le i \le k,
\end{cases}
\]
and extend $h$ to $A^*$ by concatenation, i.e., $h(xy)=h(x)h(y)$ for all $x,y \in A^*$.
Recall that the complement mapping $cm$ on $\Sigma_{\mathrm{DNA}}$
is defined by
$cm(\texttt{A}) = \texttt{T}$ and
$cm(\texttt{C}) = \texttt{G}$.
By construction, we have $\op{rc}(h(a_x)) = h(cm(a_x))$ for every $a_x \in A$.
Consequently, $\op{rc}(h(s)) = h(\op{rc}(s))$ holds for every string $s \in A^*$.
Moreover, for any distinct symbols $a_x, a_y \in A$ with $x \neq y$,
the strings $h(a_x)$ and $h(a_y)$ do not overlap.
The only overlap of $h(a_x)$ with itself is the trivial one of its full length.
Finally, note that $|h(a_x)| = 2(k+1-j)$ for all $a_x \in A$; hence,
$|h(s)| = 2(k+1-j)\,|s|$ for every string $s \in A^*$.

Now consider the instance
$R = \{h(s) \mid s \in S\}$ over $\Sigma_{\mathrm{DNA}}^*$.
Let $r$ be an arbitrary approximate solution to the SCS-RC instance $R$,
which can be represented by a permutation $(i_1, \dots, i_m)$ of $\{1, \dots, m\}$ as
$r = \langle h(s'_{i_1}), \dots, h(s'_{i_m}) \rangle$.
By the non-overlapping property of the encoding $h$, as shown above,
that is, overlaps can occur only at the boundaries of the blocks $h(a_i)$, $r$ can be transformed into an approximate solution
$s = \langle s'_{i_1}, \dots, s'_{i_m} \rangle$
to the original instance $S$, and we have
$|r| = 2(k+1-j)\,|s|$.
This implies that
\[
\op{OPT}(R) = 2(k+1-j)\,\op{OPT}(S).
\]
Moreover, let $\op{ALG}(R)$ denote the length of a $\beta$-approximate solution
to the instance $R$.
By the above argument, this solution can be transformed into a solution
to the original instance $S$ of length $\op{ALG}(R) / 2(k+1-j)$.
Therefore, we obtain
\[
\frac{\op{ALG}(R) / 2(k+1-j)}{\op{OPT}(S)}
= \frac{\op{ALG}(R)}{\op{OPT}(R)}
\le \beta,
\]
which shows that a $\beta$-approximate solution for $R$
yields a $\beta$-approximate solution for $S$.

\section{Conclusion and Future Work}

In this paper, we established an $\frac{8}{3}$-approximation algorithm for the SCS-RC problem and proved several hardness results.
Our algorithm is based on computing an optimal constrained cycle cover on the distance graph, which can be viewed as a special case of the assignment-based relaxation for the generalized traveling salesman problem (GTSP), where the requirement of forming a single tour is relaxed to allow multiple disjoint cycles.

It is well known that the assignment-based relaxation of the classical traveling salesman problem reduces to finding a minimum-weight cycle cover in the standard sense, which is solvable in polynomial time.
In contrast, the assignment-based relaxation of the GTSP is NP-hard, as shown in \cite{noon1988generalized}, where the hardness follows from a reduction to three-dimensional matching and involves clusters of size $\Theta(|V|^{\frac{2}{3}})$.

Our result demonstrates that the problem studied in this paper constitutes a nontrivial special case of the GTSP assignment-based relaxation that nevertheless admits a polynomial-time optimal algorithm.
Identifying more precise structural conditions that separate tractable cases from intractable ones remains an important direction for future work.

Another important direction for future work is to further improve the approximation guarantee.
One possible approach is to improve the compression ratio of the compression subroutine
used in greedy-based algorithms~\cite{YamanoShibuya2026SCSRC}.
Here, the compression ratio is defined as the ratio between the reduction from the total input length
achieved by an approximate solution and that achieved by an optimal solution.
For the standard SCS problem, approximation algorithms for the maximum asymmetric traveling salesman problem
can be employed as black-box subroutines.
In contrast, for the SCS-RC problem, one must additionally account for clusters formed by reverse-complement pairs,
which constitutes a nontrivial extension of this approach.
A deeper understanding of how such small clusters affect both computational complexity
and approximability may lead to stronger approximation algorithms.


\bibliography{lipics-v2021-sample-article}

\end{document}